\documentclass[a4paper]{article}

\usepackage{amsmath}
\usepackage{amsfonts}
\usepackage{amssymb}
\usepackage{amsthm}
\usepackage{hyperref}

\usepackage[english]{babel}
\usepackage{xcolor}

\theoremstyle{definition}
\newtheorem{theorem}{Theorem}
\newtheorem{lemma}[theorem]{Lemma}
\newtheorem{proposition}[theorem]{Proposition}
\newtheorem{corollary}[theorem]{Corollary}
\newtheorem{example}[theorem]{Example}
\newtheorem{definition}[theorem]{Definition}

\newcommand{\shortrules}[6]{\noindent\begin{minipage}{#6ex}{\bfseries #1}\end{minipage} $\;$ #2 $\;\Rightarrow_{\text{#5}}\;$ #3 \par\smallskip\noindent #4}

\newcommand{\comp}{\operatorname{comp}} %formula
\newcommand{\mGnd}{\operatorname{gnd}} %formula
\newcommand{\mMGU}{\operatorname{mgu}} %formula
\newcommand{\dom}{\operatorname{dom}}
\newcommand{\cdom}{\operatorname{codom}}

\newcommand{\forgd}{\operatorname{fgd}} %filter forground terms, clauses
\newcommand{\backgd}{\operatorname{bgd}} %filter background terms, clauses
\newcommand{\adiff}{\operatorname{adiff}} %filter background terms, clauses

\newcommand{\SCL}{{\text{SCL}}}
\newcommand{\SCLT}{{\text{\SCL(T)}}}

\theoremstyle{definition}

\newcommand{\ie}{i.e.\ }

\newcounter{typonote}

\newcounter{sidenote}

\newcommand{\MEMO}[1]{}
\newcommand{\con}{\operatorname{con}}

\newcommand{\sig}{\Sigma}
\newcommand{\opers}{\Omega}
\newcommand{\preds}{\Pi}
\newcommand{\term}{T}

\newcommand{\sigval}{\mathcal{A}}

\newcommand{\theory}{\mathcal{T}}
\newcommand{\sorts}{\mathcal{S}}

\newcommand{\hspec}{\mathcal{H}}
\newcommand{\bspec}{\mathcal{B}}
\newcommand{\fspec}{\mathcal{F}}
\newcommand{\modset}{\mathcal{C}}
\newcommand{\varset}{\mathcal{X}}

\newcommand{\Rat}{\mathbb{Q}}
\newcommand{\Nat}{\mathbb{N}}

\newcommand{\pnat}{\ensuremath{\operatorname{Nat}}}
\newcommand{\vars}{\ensuremath{\operatorname{vars}}}
\newcommand{\atoms}{\ensuremath{\operatorname{atoms}}}

\newcommand{\arity}{\operatorname{arity}}

\newcommand{\LRA}{\ensuremath{\operatorname{LRA}}}
\newcommand{\BS}{\ensuremath{\operatorname{BS}}}

\title{{\SCL} with Theory Constraints}

\author{Martin Bromberger\\ Max Planck Institute for Informatics\\ Saarland Informatics Campus Saarbr\"ucken\\  \and
  Alberto Fiori\\ Max Planck Institute for Informatics\\ Saarland Informatics Campus Saarbr\"ucken and\\
  Graduate School of Computer Science Saarbr\"ucken, Germany\\ \and
  Christoph Weidenbach \\ Max Planck Institute for Informatics\\ Saarland Informatics Campus, Saarbr\"ucken
	Germany\\}

\begin{document}

\maketitle

\begin{abstract}
We lift the SCL calculus for first-order logic without equality
to the SCL(T) calculus
for first-order logic without equality modulo a background theory. 
In a nutshell, the SCL(T) calculus describes a new way to guide hierarchic resolution inferences 
by a partial model assumption instead of an \textit{a priori} fixed order as done for instance in hierarchic superposition. 
The model representation consists of ground background theory literals
and ground foreground first-order literals.
One major advantage of the model guided approach is that clauses generated by SCL(T) enjoy a non-redundancy property
that makes expensive testing for tautologies and forward subsumption completely obsolete. 
SCL(T) is a semi-decision procedure for pure clause sets that
are clause sets without first-order function symbols ranging into the
background theory sorts. Moreover, SCL(T) can be turned into a decision
procedure if the considered combination of a first-order logic modulo
a background theory enjoys an abstract finite model property. 
\end{abstract}

\section{Introduction}

The combination of first-order logic reasoning with theories is still a challenge. More specifically,
there are many open problems left if first-order clauses containing universally quantified variables
are combined with theories such as (linear) arithmetic. While the ground case of such a combination
is known to be solvable by SMT~\cite{NieuwenhuisEtAl06},
hierarchic superposition~\cite{BachmairGanzingerEtAl94,kruglov2012superposition},
the sequent calculus~\cite{Ruemmer08}, or
model evolution~\cite{BaumgartnerEtAl08},
the existence of universally quantified variables poses a number of new challenges.

For instance, the combination of linear rational arithmetic (LRA) with the Bernays Schoenfinkel fragment of first-order
logic ($\BS(\LRA)$) is already undecidable if we allow a single, free constant ranging into the arithmetic sort.~\cite{Downey1972,HorbachEtAl17ARXIV}.
This can be shown by encoding the halting problem of a 2-counter machine~\cite{Minsky67}.
In Section~\ref{sec:prelim} we show another example, Example~\ref{exam:noncompalpure}, where we force a predicate to contain exactly the natural numbers.
The natural numbers cannot be defined in the first-order theory of LRA without extra first-order predicates, nor in pure first-order logic.
For a number of universally quantified fragments there
exist complete methods and some fragments are even decidable~\cite{GeMoura09,kruglov2012superposition,Voigt17Frocos}. In addition, completeness
for undecidable fragments can be obtained. For example, if the first-order part of $\BS(\LRA)$
consists only of variables and predicates, then hierarchic superposition is refutationally complete~\cite{BachmairGanzingerEtAl94}.
Already the addition of a single, free constant ranging into the arithmetic sort may destroy refutational completeness,
if such a constant cannot be mapped to a background theory domain constant, see Example~\ref{exam:noncompalpure}.

In this paper we introduce a new calculus $\SCL(T)$ for the combination of a background theory
with a foreground first-order logic without equality. As usual in a hierarchic setting, we assume the
background theory to be term-generated and compact. In this paper we only consider \emph{pure} clause sets where the
only symbols occurring in the clause set from the foreground logic are predicates and variables. As a running example,
we present the combination of linear rational arithmetic (LRA) with the Bernays Schoenfinkel fragment of first-order
logic ($\BS(\LRA)$) where for simplicity we only consider constants from LRA. In Section~\ref{sec:prelim}, we discuss that this does
not restrict expressiveness.

%\Todo{Remove this paragraph.}
For example, for a clause set out of $\BS(\LRA)$ where we write clauses in a constraint form $\Lambda\| C$ where $\Lambda$ is a conjunction of $\LRA$ atoms
of the form $x = k\cdot y$ for $k\in\Nat$ and $\LRA$ sort variables $x,y$, or atoms $x>0$, and the first-order part $C$ only
contains a single monadic predicate $P$ applied to some variable, unsatisfiability is already undecidable~\cite{Downey1972,HorbachEtAl17ARXIV}. This can be shown
by encoding the halting problem of a 2-counter machine~\cite{Minsky67}.

%\Todo{Remove this paragraph.}
For a further example, the two clauses\newline
\centerline{$x = 0 \| \pnat(x) \quad y = x + 1\| \neg \pnat(x) \lor \pnat(y)$}
\noindent
belong to our fragment as well. For any $\BS(\LRA)$ algebra satisfying the two clauses, the predicate $\pnat$ contains
the natural numbers. In Section~\ref{sec:prelim} we will even show that it can be forced to contain exactly the natural numbers by
adding further clauses from our fragment, Example~\ref{exam:noncompalpure}. 

In contrast to superposition based approaches to hierarchic reasoning~\cite{BachmairGanzingerEtAl94,kruglov2012superposition,BaumgartnerWaldmann19}
where inferences are restricted by an \emph{a priori} ordering, our idea is to select inferences via a partial model assumption,
similar to CDCL~\cite{MSS96,BayardoSchrag97} and respective calculi using an explicit model assumption to direct
reasoning in first-order logic~\cite{BaumgartnerFuchsTinelli06,PiskacEtAl10,Korovin13,TeuckeWeidenbach15,AlagiWeidenbach15,BonacinaPlaisted16}.
In particular, we want to lift our previously obtained
results for model-based reasoning in first-order logic~\cite{FioriWeidenbach19} to first-order logic modulo a background theory.
There, models are ground and either extended by guessing an undefined ground literal or by propagating
a new ground literal with respect to a clause and the current ground model assumption. 
The two clauses $x = 0 \| \pnat(x)$,  $y = x + 1\| \neg \pnat(x) \lor \pnat(y)$ enable already infinitely many propagations\newline
\centerline{$\pnat(0)$, $\pnat(1)$, \;$\pnat(2)$, $\ldots$.}
\noindent
Therefore, exhaustive
propagation cannot be permitted. On the other hand, exhaustive propagation guarantees that learned clauses are not redundant with respect to
an adopted superposition redundancy notion~\cite{Weidenbach15,FioriWeidenbach19}. Restricting a calculus to non-exhaustive propagation while having non-redundancy properties
for learned clauses is difficult to obtain. One reason why model-based reasoning for LIA is currently inferior compared to classical branch and bound approaches
used in SMT~\cite{NieuwenhuisEtAl06,BrombergerEtAl15} is a missing non-redundancy guarantee for inferred inequations. This is partly a result of the fact that
exhaustive propagation does not terminate in LIA as well and is therefore not permitted. Our solution in this paper is to restrict model assumptions
to finite ground models, similar to \cite{FioriWeidenbach19}. The models are build with respect to a fixed, finite set $B$ of foreground constants of the
arithmetic sort. Propagations are not exhaustive but restricted to a finite number, typically not exhausting $B$. If $B$ is not grown, the proposed calculus $\SCLT$ always
terminates. Either by finding a contradiction or finding a model with respect to $B$. Of course, this (finite) model needs not to be extendable to a model of the
clause set. Satisfiability is undecidable for pure clause sets. Still for some further restricted fragments of pure clause sets the approach can be turned
into a decision procedure, see Section~\ref{sec:extdisc}. In general, the set $B$ has to be extended during the search for a refutation. Nevertheless, we can prove
that learned clauses in $\SCL(T)$ are always non-redundant with respect to an adopted superposition redundancy notion, Lemma~\ref{lemm:non-red}. The notion includes subsumption
of constraint clauses and semantic tautologies. That means, $\SCL(T)$ will never learn a true clause, nor a clause that is subsumed by an existing clause.
Restricting the model representation language to ground literals of the background and foreground logic enables detecting false clauses or propagating clauses efficiently by SMT. 
Still, the $\SCL(T)$ calculus is sound, Lemma~\ref{lem:sclsound}, and refutationally complete, Theorem~\ref{theo:groununsatcomple}, i.e.,
a regular strategy is guaranteed to find a refutation for pure clause sets if it exists.

\paragraph{Related Work:} In contrast to variants of hierarchic superposition~\cite{BachmairGanzingerEtAl94,kruglov2012superposition,BaumgartnerWaldmann19}
$\SCLT$ selects clauses via a partial model assumption and not via an ordering. This has the advantage that $\SCLT$ does not generate redundant clauses.
It lacks the implicit model building capabilities of hierarchic superposition. Instead it generates finite models candidates that need to be extended
to overall models of the clause set, see Example~\ref{exa:modelextract}.
One way to deal with universally quantified variables in an SMT
setting is via instantiation~\cite{GeMoura09,ReynoldsEtAl18}. This has shown to be practically useful in many applications. It typically comes without completeness
guarantees and it does not learn any new constrained clauses. An alternative is to combine SMT techniques with superposition~\cite{MouraB08} where the ground
literals from an SMT model assumption are resolved by superposition with first-order clauses. $\SCL(T)$ does not resolve with respect to its ground model
assumption but on the original clauses with variables. Background theories can also be build into first-order superposition in a kind of lazy way. This direction
has been followed by SPASS+T~\cite{PrevostoWaldmann06} and Vampire~\cite{KovacsVoronkov13}. The idea is to axiomatize part of the background theory in
first-order logic and to direct ground literals of the background theory to SMT solver. Also this approach has shown to be practically useful but comes
without any completeness guarantees and generated clauses may be redundant. Model evolution~\cite{BaumgartnerFuchsTinelli06} has also been extended
with linear integer arithmetic~\cite{BaumgartnerEtAl08} where universally quantified integer variables are finitely bound from the beginning.
A combination of first-order logic with linear integer arithmetic has also been built into a sequent calculus~\cite{Ruemmer08} that operates
in the style of a free-variable tableau calculus with incremental closure. No new clauses are learned.

\paragraph{Organization of the Paper:} After a section fixing notation, notions and some preliminary work, Section~\ref{sec:prelim},
the following Section~\ref{sec:sclrulesprop} introduces the $\SCLT$ calculus and proves its properties. The final Section~\ref{sec:extdisc}
discusses extensions to model building, further improvements and summarizes the obtained results.

%%% Local Variables:
%%% mode: latex
%%% TeX-master: paper.tex
%%% End:

\section{Preliminaries} \label{sec:prelim}

\paragraph{Many-Sorted First-Order Logic without Equality:} A \emph{many-sorted signature} $\sig = (\sorts,\opers, \Pi)$ is a triple consisting of a finite, non-empty 
set $\sorts$ of \emph{sort symbols}, a non-empty set
$\opers$ of \emph{operator symbols} (also called \emph{function symbols}) over $\sorts$ and a finite set 
$\Pi$ of \emph{predicate symbols} over $\sorts$.
For every sort from $\sorts$ there is at least one constant symbol in $\opers$ of this sort.
First-order terms, atoms, literals, clauses, formulas and substitutions are defined in the
usual many-sorted way where an additional infinite set $\varset$ of variables is assumed, such that
for each sort from $\sorts$ there are infinitely many variables of this sort in  $\varset$.
For each sort $S\in\sorts$, $\term_{S}(\sig,\varset)$ denotes the set of all terms of sort $S$
and $\term_{S}(\sig)$ the set of all ground terms of sort $S$.

For notation, $a, b, c$ are constants from $\opers$,
$w, x, y, z$ variables from $\varset$, and if we want to emphasize the sort
of a variable, we write $x_S$ for a variable of sort $S$;
$t, s$ denote terms, $P, Q, R$ predicates from $\preds$, $A, B$ atoms,
$L, K, H$ denote literals, $C, D$ denote clauses, and $N$ denotes a clause set. For substitutions
we write $\sigma, \delta, \rho$. Substitutions are well-sorted: if $x_s\sigma = t$ then $t\in \term_{S}(\sig,\varset)$,
they have a finite domain $\dom(\sigma) = \{ x \mid x\sigma \neq x\}$ and their codomain is denoted by $\cdom(\sigma) = \{x\sigma\mid x\in\dom(\sigma)\}$.
The application of substitutions is homomorphically extended to non-variable terms, atoms, literals, clauses, and formulas.
The complement of a literal is denoted by the function $\comp$.
For a literal $L$, $|L|$ denotes its respective atom. The function $\atoms$ computes the
set of atoms from a clause or clause set.
The function $\vars$ maps terms, literals, clauses to their respective set of contained variables.
The function $\con$ maps terms, literals, clauses to their respective set of constants.
A term, atom, clause, or a set of these objects is \emph{ground}
 if it does not contain any variable, i.e., the function $\vars$ returns the empty set.
A substitution $\sigma$ is \emph{ground} if $\cdom(\sigma)$ is ground.
A substitution $\sigma$ is \emph{grounding} for a term $t$, literal $L$,
 clause $C$ if $t\sigma$, $L\sigma$, $C\sigma$ is ground, respectively.
The function $\mGnd$ computes the set of all ground instances of
 a literal, clause, or clause set. Given a set of constants $B$,
 the function $\mGnd_B$ computes the set of all ground instances
 of a literal, clause, or clause set where the grounding is restricted to use constants from $B$.
 The function $\mMGU$ denotes the \emph{most general unifier} of two terms, atoms, literals.
 As usual, we assume that any $\mMGU$ of two terms or literals does not introduce any fresh variables and is idempotent.

The semantics of many-sorted first-order logic is given by the notion of an algebra:
let $\sig = (\sorts, \opers, \preds)$ be a many-sorted signature. 
A \emph{$\sig$-algebra} $\sigval$, also called \emph{$\sig$-interpretation}, is a mapping that assigns
(i)~a non-empty carrier set $S^{\sigval}$ to every sort $S \in \sorts$, 
so that $(S_1)^{\sigval} \cap (S_2)^{\sigval} = \emptyset$ for any distinct sorts $S_1, S_2 \in \sorts$,
(ii)~a total function $f^{\sigval} : (S_1)^{\sigval} \times \ldots \times (S_n)^{\sigval} \rightarrow (S)^{\sigval}$ to 
every operator $f \in \opers$, $\arity(f)=n$ where $f: S_1 \times \ldots \times S_n \rightarrow S$,
(iii)~a relation $P^{\sigval} \subseteq ((S_1)^{\sigval} \times \ldots \times (S_m)^{\sigval})$ to every predicate symbol $P \in \Pi$ with $\arity(P)=m$.
The semantic entailment relation $\models$ is defined in the usual way. 
We call a $\sig$-algebra $\sigval$ \emph{term-generated} if 
$\sigval$ fulfills the following condition: whenever $\sigval$ entails all groundings $C\sigma$ of a clause $C$ (i.e., $\sigval \models C\sigma$ for all grounding substitutions $\sigma$ of a clause $C$), then $\sigval$ must also entail $C$ itself (i.e., $\sigval \models C$).
%We call a $\sig$-algebra $\sigval$ \emph{term-generated} if
%for all clauses $C$ if whenever for all grounding substitutions $\sigma$ of $C$ we have $\sigval \models C\sigma$, then $\sigval \models C$.

\paragraph{Hierarchic Reasoning:} Starting point of a hierarchic reasoning~\cite{BachmairGanzingerEtAl94,BaumgartnerWaldmann19} is a background
theory $\theory^{\bspec}$ over a many-sorted signature  $\sig^{\bspec} = (\sorts^{\bspec},\opers^{\bspec}, \preds^{\bspec})$
and a non-empty set of term-generated $\sig^{\bspec}$-algebras $\modset^{\bspec}$: $\theory^{\bspec} = (\sig^{\bspec}, \modset^{\bspec})$.
A constant $c\in\opers^{\bspec}$ is called a \emph{domain constant} if $c^{\sigval}\neq d^{\sigval}$ for all $\sigval\in\modset^{\bspec}$ and for all $d\in\opers^{\bspec}$ with $d\neq c$.
The background theory is then extended via a foreground signature $\sig^{\fspec} = (\sorts^{\fspec},\opers^{\fspec}, \preds^{\fspec})$
where $\sorts^{\bspec}\subseteq \sorts^{\fspec}$, $\opers^{\bspec}\cap \opers^{\fspec} = \emptyset$, and $\preds^{\bspec}\cap \preds^{\fspec} = \emptyset$.
Hierarchic reasoning is based on a background theory $\theory^{\bspec}$ and
a respective foreground signature $\sig^{\fspec}$: $\hspec = (\theory^{\bspec},  \sig^{\fspec})$. It has its
associated signature $\sig^{\hspec} = (\sorts^{\fspec}, \opers^{\bspec}\cup\opers^{\fspec}, \preds^{\bspec}\cup \preds^{\fspec})$
generating \emph{hierarchic}  $\sig^{\hspec} $-algebras. A $\sig^{\hspec}$-algebra $\sigval$ is called \emph{hierarchic} with respect to its background theory $\theory^{\bspec}$,
if $\sigval^{\hspec}|_{\sig^{\bspec}} \in \modset^{\bspec}$. As usual, $\sigval^{\hspec}|_{\sig^{\bspec}}$ is obtained from a $\sigval^{\hspec}$-algebra
by removing all carrier sets $S^{\sigval}$ for all $S\in(\sorts^{\fspec}\setminus\sorts^{\bspec})$, all functions from $\opers^{\fspec}$
and all predicates from $\preds^{\fspec}$.
We write $\models_\hspec$ for the entailment relation with respect to hierarchic algebras and formulas from  $\sig^{\hspec}$
and $\models_\bspec$ for the entailment relation with respect to the $\modset^{\bspec}$ algebras and formulas from $\sig^{\bspec}$.

Terms, atoms, literals build over $\sig^{\bspec}$ are called \emph{pure background terms}, 
\emph{pure background atoms}, and
\emph{pure background literals}, respectively.
All terms, atoms, with a top-symbol from $\opers^{\bspec}$ or $\preds^{\bspec}$, respectively, are called \emph{background terms}, \emph{background atoms}, respectively. 
A background atom or its negation is a \emph{background literal}. 
All terms, atoms, with a top-symbol from $\opers^{\fspec}$ or $\preds^{\fspec}$, respectively, are called \emph{foreground terms}, \emph{foreground atoms}, respectively. 
A foreground atom or its negation is a \emph{foreground literal}.
Given a set (sequence) of $\hspec$ literals, the function $\backgd$ returns the set (sequence) of background literals
and the function $\forgd$ the respective set (sequence) of foreground literals.
A substitution $\sigma$ is called \emph{simple} if $x_S\sigma\in\term_{S}(\sig^{\bspec},\varset)$ for
all $x_S\in\dom(\sigma)$ and $S\in\sorts^{\bspec}$.

As usual, clauses are disjunctions of literals with implicitly universally quantified variables. We often write a $\sig^{\hspec}$ clause
as a \emph{constrained clause}, denoted $\Lambda\parallel C$ where $\Lambda$ is a conjunction of background literals
and $C$ is a disjunction of foreground literals semantically denoting the clause $\neg \Lambda\lor C$.
A \emph{constrained closure}  is denoted as $\Lambda\parallel C\cdot\sigma$ where $\sigma$ is grounding for $\Lambda$ and $C$.
A constrained closure $\Lambda\parallel C\cdot\sigma$ denotes the ground constrained clause  $\Lambda\sigma \parallel C\sigma$.

In addition, we assume a well-founded, total, strict ordering $\prec$ on ground literals, called an $\hspec$-order, such that
 background literals are smaller than foreground literals.
This ordering is then lifted to constrained clauses and sets thereof by its respective multiset extension.
We overload $\prec$ for literals, constrained clauses, and sets of constrained clause
 if the meaning is clear from the context.
We define $\preceq$ as the reflexive closure of $\prec$ and
$N^{\preceq \Lambda\parallel C} := \{D \mid D\in N \;\text{and}\; D\preceq \Lambda\parallel C\}$.
An instance of an LPO with according precedence can serve as $\prec$.
%\Todo{This makes it sound like we assume one fixed ordering. Is this intentional?}
%In addition, we need the notion of $\hspec$-orders, which are well-founded, total, strict orderings $\prec$ on ground literals such that
% background literals are smaller than foreground literals. 
%For instance, LPO with according precedence is an $\hspec$-order. 
%These orderings can then be lifted to constrained clauses and sets thereof by their respective multiset extensions.
%We overload $\prec$ for literals, constrained clauses, and sets of constrained clause
% if the meaning is clear from the context.
%We define $\preceq$ as the reflexive closure of $\prec$ and
%$N^{\preceq \Lambda\parallel C} := \{D \mid D\in N \;\text{and}\; D\preceq \Lambda\parallel C\}$.

\begin{definition}[Clause Redundancy] \label{prelim:def:redundancy}
  A ground constrained clause $\Lambda\parallel C$ is \emph{redundant} with respect to a set $N$
  of ground constrained clauses and an order $\prec$ if
  $N^{\preceq \Lambda\parallel C} \models_{\mathcal H} \Lambda\parallel C$.
  A clause $\Lambda\parallel C$ is \emph{redundant}  with respect to a clause set
  $N$, an $\hspec$-order $\prec$, and a set of constants $B$ 
  if for all $\Lambda'\parallel  C' \in \mGnd_B(\Lambda\parallel  C)$  the clause $\Lambda'\parallel  C'$ is
  redundant with respect to $\cup_{D\in N} \mGnd_B(D)$.
\end{definition}

\begin{example}[$\BS(\LRA)$]
%\Todo{We might need to adjust this if we focus more on Marco's theories!}
The running example in this paper is the Bernays-Schoen\-finkel clause fragment over linear arithmetic: $\BS(\LRA)$.
The background theory is linear rational arithmetic
over the many-sorted signature $\sig^{\LRA} = (\sorts^{\LRA}, \opers^{\LRA}, \preds^{\LRA})$ with
$\sorts^{\LRA} = \{\LRA\}$, $\opers^{\LRA} = \{0,1,+,-\}\cup\Rat$, $\preds^{\LRA} = \{\leq, <, \not =, =, >, \geq\})$
where $\LRA$ is the linear arithmetic sort, the function symbols consist of $0,1,+,-$ plus the rational numbers and
predicate symbols $\le, <, =, \not =, >, \ge$. The linear arithmetic theory $\theory^{\LRA} = (\sig^{\LRA},\{\sigval^{\LRA}\})$
consists of  the linear arithmetic signature together with the standard model $\sigval^{\LRA}$ of linear arithmetic.
This theory is then extended by the free (foreground) first-order signature
$\sig^{\BS} = (\{\LRA\},\opers^{\BS},\preds^{\BS})$ where $\opers^{\BS}$ is a set of constants of sort $\LRA$ different from $\opers^{\LRA}$ constants,
and $\preds^{\BS}$ is a set of first-order predicates over the sort $\LRA$. We are interested in hierarchic algebras $\sigval^{\BS(\LRA)}$ over
the signature $\sig^{\BS(\LRA)} = (\{\LRA\},\opers^{\BS}\cup \opers^{\LRA},\preds^{\BS}\cup \preds^{\LRA})$ that are $\sig^{\BS(\LRA)}$ algebras
such that $\sigval^{\BS(\LRA)}|_{\sig^{\LRA}} = \sigval^{\LRA}$.
\end{example}

Note that our definition of the $\BS(\LRA)$ fragment restricted to the linear arithmetic sort does not restrict
expressiveness compared to a definition adding further free sorts to $\sig^{\BS}$. Free sorts containing only constants
can be simulated by the linear arithmetic
sort in a many-sorted setting. For example, assume a free sort $S$ with constants $k_1,\ldots,k_n$, then this sort can be represented
by the $\LRA$ sort with domain constants $1,\ldots, n$ where each occurrence of a $k_i$ in a clause set is replaced by $i$ and each occurrence of a variable
$u_S$ with a fresh variable $x$ of sort $\LRA$ plus the additional constraint $P(x)$. The fresh predicate $P$ is true
exactly for the numbers $1,\ldots,n$. This is encoded by the
first-order $\sig^{\BS(\LRA)}$ formula $\forall x [P(x) \leftrightarrow (x= 1\lor\ldots\lor x= n)]$.
For example, a clause $y>1 \parallel  R(y,k_5,u_P)$ is encoded by the clause $y>1\parallel  \neg P(x) \lor R(y,5,x)$ preserving satisfiability.
Please note that this encoding does not introduce any free constants of the $\LRA$ sort. The restriction to background sorts
simplifies the presentation of the $SCL(T)$ calculus in Section~\ref{sec:sclrulesprop} significantly.

We call a clause set $N$ \emph{abstracted} if the arguments of any predicate from $\preds^{\fspec}$ are
only variables. Abstraction can always be obtained by adding background constraints, e.g., the $\BS(\LRA)$
clause $x >1 \parallel  R(x, 5)$ can be abstracted to $x >1, y = 5 \parallel  R(x, y)$, preserving satisfiability.
Recall that even in the foreground signature we only consider background sorts and that the only
operators in the foreground signature are constants.

A set $N$ of $\hspec$ clauses is called \emph{pure} if it does not contain symbols from $\opers^{\fspec}$
ranging into a sort of $\sorts^\bspec$. In this case $N$
is \emph{sufficiently complete} according to~\cite{BachmairGanzingerEtAl94}, hence hierarchic superposition
is complete for $N$~\cite{BachmairGanzingerEtAl94,BaumgartnerWaldmann19}.
As a consequence, a pure clause set $N$ is unsatisfiable iff $\mGnd_B(N)$ is unsatisfiable for a sufficiently
large set $B$ of constants. We will make use of this result in the completeness proof for our calculus, Theorem~\ref{theo:groununsatcomple}.
A set $N$ of $\hspec$ clauses is called \emph{almost pure}, if the only symbols it contains
from $\opers^{\fspec}$ ranging into a sort of $\sorts^\bspec$ are constants. In general, there cannot be a complete
calculus for almost pure clause sets, see Example~\ref{exam:noncompalpure}. Satisfiability of pure clause sets
is undecidable. We already mentioned in the introduction that this can be shown through a reduction to the
halting problem for two-counter machines~\cite{Minsky67,HorbachEtAl17ARXIV}. Clause redundancy for pure
clause sets cannot be decided as well.

\begin{lemma}[Non-Redundancy for Pure Clause Sets is Undecidable]
For a pure clause set $N$ it is undecidable whether some clause $C$ is non-redundant with respect to $N$.
\end{lemma}
\begin{proof}
  The construction is similar to the construction taken in \cite{Weidenbach15}.
  Let $N=\{C_1,\ldots,C_n\}$ be an arbitrary pure clause set. Then $N$ is satisfiable iff $N' = \{P\lor C_1, Q\lor C_2, C_3, \ldots, C_n, \neg P, \neg Q\}$ is satisfiable
  for two fresh propositional variables (predicates of arity zero)  $P, Q$. Then $N'$ is satisfiable iff $P\lor Q$ is non-redundant with respect to $N'\setminus\{\neg P, \neg Q\}$ where we
  assume that $P,Q$ are maximal in the atom ordering $\prec$.
\end{proof}

\begin{example}[Non-Compactness of Almost Pure $\BS(\LRA)$ Clause Sets] \label{exam:noncompalpure}
Note that for $\BS(\LRA)$ clauses sets that are not pure, i.e., they contain $\sig^{\BS}$ constants,
there cannot be a complete calculus, in general, because compactness~\cite{BachmairGanzingerEtAl94} is lost. For example, consider
the five abstracted clauses\newline
\centerline{$\begin{array}{c}
\pnat(0),\; y = x + 1\parallel  \neg \pnat(x) \lor \pnat(y), \\
x < 0 \parallel  \neg \pnat(x),\;  0 < x < 1 \parallel  \neg \pnat(x),\; x > 0, y = x + 1 \parallel  \neg \pnat(y) \lor \pnat(x)
            \end{array}$}
          
              \noindent
defining the natural numbers with respect to $\theory^{\LRA}$ and a foreground predicate $\pnat$. The four clauses\newline
\centerline{$x = a \parallel  \pnat(x),\; P(0),\; y = x+ 1 \parallel  \neg P(x) \lor P(y),\; x = a \parallel  \neg P(x)$}

\noindent
express that $a$ is a natural number but not contained in $P$ that itself contains all natural numbers.
The clause set is unsatisfiable, but there is no proof by hierarchic superposition (resolution), because
any finite set of ground clauses out of this clause set is satisfiable. If the semantics of $\BS(\LRA)$ clause
sets is not restricted to hierarchic algebras, there is a model for the clause set by adding some junk element to $\Rat$
and assigning it to $a$.
\end{example}

%%% Local Variables:
%%% mode: latex
%%% TeX-master: "paper"
%%% End:

\section{$\SCLT$} \label{sec:sclrulesprop}

\paragraph{\bfseries Assumptions:}
For this section we consider only pure, abstracted clause sets $N$. We assume that
the background theory $\theory^{\bspec}$ is term-generated, compact, contains
an equality $=$, and that all constants of the background signature are domain constants.
We further assume that the set $\opers^{\fspec}$ contains infinitely
many constants for each background sort.

\begin{example}[Pure Clauses]
  With respect to $\BS(\LRA)$ the unit clause $x \geq 5, 3x + 4y = z \parallel Q(x,y,z)$ is abstracted and pure
  while the clause $x \geq 5, 3x + 4y = a, z = a \parallel Q(x,y,z)$ is abstracted but not pure because of the foreground constant $a$
  of the $\LRA$ sort,
  and the clause $x \geq 5, 3x + 4y = 7 \parallel Q(x,y,7)$ is not abstracted.
\end{example}

\bigskip
Note that for pure, abstracted clause sets, any unifier between two foreground
literals is simple and its codomain consists of variables only.

In order for the $\SCLT$ calculus to be effective, decidability in
$\theory^{\bspec}$ is needed as well. For the calculus we implicitly use
the following equivalence: A $\sig^{\bspec}$ sentence
$\exists x_1,\ldots,x_n \phi$ where $\phi$ is quantifier free is
true, i.e., $\models_{\bspec} \exists x_1,\ldots,x_n \phi$ iff the ground
formula $\phi\{x_1\mapsto a_1,\ldots, x_n\mapsto a_n\}$ where the $a_i$
are $\opers^{\fspec}$ constants of the respective background sorts
is $\hspec$ satisfiable. Together with decidability in $\theory^{\bspec}$
this guarantees decidability of the satisfiability of ground constraints from
constrained clauses.

If not stated otherwise, satisfiability means satisfiability with respect to $\hspec$.
The function $\adiff(B)$ for some finite sequence of background sort constants denotes
a constraint that implies different interpretations for the constants in $B$. In
case the background theory enables a strict ordering $<$ as $\LRA$ does,
then the ordering can be used for this purpose.
For example, $\adiff([a,b,c])$ is then the constraint $a < b < c$.
In case the background theory does not enable a strict ordering, then
inequations can express disjointness of the constants.
For example, $\adiff([a,b,c])$ is then constraint $a \neq b \land a \neq c \land b\neq c$.
An ordering constraint has the advantage over an inequality constraint that
it also breaks symmetries. Assuming all constants to
be different will eventually enable a satisfiability test for foreground literals
based on purely syntactic complementarity.

The inference rules of $\SCLT$ are represented by an abstract rewrite system.
They operate on a problem state, a six-tuple $\Gamma = (M; N; U; B; k; D)$
where $M$ is a sequence of annotated ground literals,
the \emph{trail}; $N$ and $U$ are the sets of \emph{initial} and
\emph{learned} constrained clauses; $B$ is a finite sequence of constants of background sorts
for instantiation; $k$ counts the number
of decisions in $M$; and $D$ is a constrained closure that is either $\top$, $\Lambda\parallel\bot\cdot\sigma$,
or $\Lambda\parallel C\cdot\sigma$. Foreground literals in $M$ are either
annotated with a number, a level; \ie, they have the form $L^k$ meaning that
$L$ is the $k$-th guessed decision literal, or they are annotated with a
constrained closure that propagated the literal to become true, \ie, they have
the form $(L\sigma)^{(\Lambda\parallel C\lor L)\cdot\sigma}$. An annotated literal is called a decision literal if it is of the
form $L^k$ and a propagation literal or a propagated literal if it of in the form
$L\cdot\sigma^{(\Lambda\parallel C\lor L)\cdot\sigma}$.
A ground foreground literal $L$ is of \emph{level} $i$ with respect to a problem state $(M; N; U; B; k; D)$ if $L$ or
$\comp(L)$ occurs in $M$ and the first decision literal left from $L$
($\comp(L)$) in $M$, including $L$, is annotated with $i$. If there is no such
decision literal then its level is zero. A ground constrained clause $\Lambda\parallel C$
is of \emph{level} $i$ with respect to a problem state
$(M; N; U; B; k; D)$ if $i$ is the maximal level of a foreground literal in $C$; the level
of an empty clause $\Lambda\parallel \bot\cdot\sigma$ is 0.
A ground literal $L$ is \emph{undefined} in $M$
if neither $L$ nor $\comp(L)$ occur in $M$. The initial state
for a first-order, pure, abstracted $\hspec$ clause set $N$ is $(\epsilon; N; \emptyset;B; 0; \top)$,
where $B$ is a finite sequence of foreground constants of background sorts. These constants cannot occur in $N$, because $N$ is pure.
The final state $(\epsilon; N; U; B; 0; \Lambda\parallel \bot)$ denotes unsatisfiability of $N$.
Given a trail $M$ and its foreground literals $\forgd(M) = \{L_1,\ldots,L_n\}$
an  $\hspec$ ordering $\prec$ \emph{induced} by $M$ is any
$\hspec$ ordering where $L_i \prec L_j$ if $L_i$ occurs left from $L_j$ in $M$,
or, $L_i$ is defined in $M$ and $L_j$ is not.

The transition rules for $\SCLT$ are

\bigskip
\shortrules{Propagate}
{$(M;N;U;B;k;\top)$}
{$(M,L\sigma^{(\Lambda\parallel  C_0\lor L)\delta\cdot\sigma},\Lambda'\sigma;N;U;B;k;\top)$}
{provided $\Lambda\parallel  C\in (N\cup U)$,
  $\sigma$ is grounding for $\Lambda\parallel  C$,
  $\adiff(B)\land\backgd(M)\land\Lambda\sigma$
  is satisfiable,
  $C = C_0\lor C_1 \lor L$,
  $C_1\sigma = L\sigma\lor\ldots\lor L\sigma$,
  $C_0\sigma$ does not contain $L\sigma$,
  $\delta$ is the mgu of the literals in $C_1$ and $L$,
  $\Lambda'\sigma$ are the background literals from $\Lambda\sigma$ that are not yet on the trail,
  $\forgd(M)\models\lnot(C_0\sigma)$, $\cdom(\sigma)\subseteq B$, and $L\sigma$ is
  undefined in $M$
}{\SCLT}{12}

\bigskip
The rule Propagate applies exhaustive factoring to the propagated literal with
respect to the grounding substitution $\sigma$ and annotates the factored clause
to the propagation. By writing $M,L\sigma^{(\Lambda\parallel  C_0\lor L)\delta\cdot\sigma},\Lambda'\sigma$
we denote that all background literals from $\Lambda'\sigma$ are added to the trail.

\bigskip
\shortrules{Decide}
{$(M;N;U;B;k;\top)$}
{$(M,L\sigma^{k+1},\Lambda\sigma;N;U;B;k+1;\top)$}
{provided $L\sigma$ is undefined in $M$,
  $|L\sigma| \in \atoms(\mGnd_B(N\cup U))$,
  $|K\sigma| \in \atoms(\mGnd_B(N\cup U))$ for all $K\sigma\in\Lambda\sigma$,
  $\sigma$ is grounding for $\Lambda$,
  all background literals in $\Lambda\sigma$ are undefined in $M$,
  $\adiff(B)\land\backgd(M)\land\Lambda\sigma$ is satisfiable,
  and $\cdom(\sigma)\subseteq B$
}{\SCLT}{12}

% \bigskip
% \shortrules{Decide2}
% {$(M;N;U;B;k;\top)$}
% {$(M,L\sigma^{k+1},\Lambda\sigma;N;U;B;k+1;\top)$}
% {provided $L\sigma$ is undefined in $M$,
%   $|L\sigma| \in \atoms(\mGnd_B(N\cup U))$,
%   $|K\sigma| \in \atoms(\mGnd_B(N\cup U))$ for all $K\sigma\in\Lambda\sigma$,
%   $\sigma$ is grounding for $\Lambda$,
%   all background literals in $\Lambda\sigma$ are undefined in $M$,
%   $\adiff(B\cup B')\land\backgd(M)\land\Lambda\sigma$ is satisfiable
%   where $B'$ is the set of all background constants in $N$,
%   and $\cdom(\sigma)\subseteq B$
% }{\SCLT}{12}

\bigskip
Making sure that no duplicates of background literals occur on the trail by rules Propagate and Decide together
with a fixed finite sequence $B$ of constants and the restriction of Propagate and Decide to undefined literals
guarantees that the number of potential trails of a run is finite. Requiring the constants from $B$ to
be different by the $\adiff(B)$ constraint enables a purely syntactic consistency check for foreground literals.

\bigskip
\shortrules{Conflict}
{$(M;N;U;B;k;\top)$}
{$(M;N;U;B;k;\Lambda\parallel  D\cdot\sigma)$}
{provided $\Lambda\parallel D\in (N\cup U)$,
  $\sigma$ is grounding for $\Lambda\parallel D$,
  $\adiff(B)\land\backgd(M)\land\Lambda\sigma$ is satisfiable,
  $\forgd(M)\models\lnot(D\sigma)$, and $\cdom(\sigma)\subseteq B$
}{\SCLT}{12}

\bigskip
\shortrules{Resolve}
{$(M, L\rho^{\Lambda\parallel  C\lor L\cdot\rho};N;U;B;k;
  (\Lambda'\parallel  D\lor L')\cdot\sigma)$}
{$(M, L\rho^{\Lambda\parallel  C\lor L\cdot\rho};N;U;B;k;(\Lambda\land\Lambda'\parallel  D\lor C)\eta\cdot\sigma\rho)$}
{provided $L\rho = \comp(L'\sigma)$, and $\eta=\mMGU(L,\comp(L'))$}
{\SCLT}{11}

\bigskip
Note that Resolve does not remove the literal $L\rho$ from the trail. This is needed if the clause
$D\sigma$ contains further literals complementary of $L\rho$ that have not been factorized.

\bigskip
\shortrules{Factorize}
{$(M;N;U;B;k;(\Lambda\parallel  D\lor L \lor L')\cdot\sigma)$}
{$(M;N;U;B;k;(\Lambda\parallel  D\lor L)\eta\cdot\sigma)$}
{provided $L\sigma = L'\sigma$, and $\eta=\mMGU(L,L')$}
{\SCLT}{11}

\bigskip
Note that Factorize is not limited with respect to the trail. It may apply
to any two literals that become identical by application of the grounding substitution $\sigma$.

\bigskip
\shortrules{Skip}
{$(M,L;N;U;B;k;\Lambda'\parallel D\cdot\sigma)$}
{$(M;N;U;B;l;\Lambda'\parallel D\cdot\sigma)$}
{provided $L$ is a foreground literal
  and $\comp(L)$ does not occur in $D\sigma$,
  or $L$ is a background literal; if $L$ is a foreground decision literal then $l=k-1$, otherwise $l=k$
}{\SCLT}{11}

\bigskip
Note that Skip can also skip decision literals. This is needed because we won't eventually
require exhaustive propagation. While exhaustive propagation in CDCL is limited to the
number of propositional variables, in the context of our logic, for example $\BS(\LRA)$, it is exponential in the
arity of foreground predicate symbols and can lead to an unfair exploration of the space
of possible inferences, harming completeness, see Example~\ref{exa:exhaustiveprop}.

\bigskip
\shortrules{Backtrack}
{$(M,K^{i+1},M';N;U;B;k;(\Lambda\parallel  D\lor L)\cdot\sigma)$}
{$(M,L\sigma^{(\Lambda\parallel  D\lor L)\cdot\sigma},\Lambda'\sigma;N; U\cup\{\Lambda\parallel  D\lor L\};B;i;\top)$}
%{$(M,\Lambda\sigma,L\sigma^{\Lambda\parallel D\lor L\cdot\sigma};N;
%  U\cup\{\Lambda\parallel  D'\lor L\};B;i;\top)$}
{provided $L\sigma$ is of level $k$, and $D\sigma$ is of level $i$, $\Lambda'\sigma$ are the background literals from $\Lambda\sigma$ that are not yet on the trail}
{\SCLT}{11}

\bigskip
The definition of Backtrack requires that if $L\sigma$ is the only literal of level $k$ in $(D\lor L)\sigma$ then additional
occurrences of $L\sigma$ in $D$ have to be factorized first before Backtrack can be applied.

\bigskip
\shortrules{Grow}
{$(M;N;U;B;k;\top)$}
{$(\epsilon;N;U;B\cup B';0;\top)$}
{ provided $B'$ is a non-empty sequence of foreground constants of background sorts distinct from
  the constants in $B$}{\SCLT}{11}

\bigskip
In case the $\adiff$ constraint is implemented by a strict ordering predicate on the
basis of the sequence $B$, it can be useful to inject the new constants $B'$ into $B\cup B'$
such that the ordering of the constants from $B$ is not changed. This can help caching background
theory results for testing trail satisfiability.

\begin{definition}
The rules Propagate, Decide, Grow,
and Conflict are called \emph{conflict search} rules and the rules Resolve, Skip,
Factorize, and Backtrack are called \emph{conflict resolution} rules.
\end{definition}

Recall that the goal of our calculus is to replace the ordering restrictions of the hierarchic
superposition calculus with a guiding model assumption. All our inferences are hierarchic superposition
inferences where the ordering restrictions are neglected.

\begin{example}[Inconsistent Trail]
  Consider a clause set $N=\{ R(x,y), x \le y\parallel  \lnot R(x,y)\lor P(x), x \ge y\parallel  \lnot R(x,y)\lor \lnot P(y)\}$;
  if we were to remove the $\adiff(B)$ constraint from the side conditions of rule Propagate
  we would be able to obtain inconsistent trails.
  Starting with just $B=\{a,b\}$ as constants it is
  possible to propagate three times and obtain the
  trail $M = [R(a,b),P(a), a\le b, \lnot P(b), a \ge b]$,
  $M$ is clearly inconsistent as $M\models P(a)$, $M\models \lnot P(b)$ yet $a=b$.
\end{example}

\begin{example}[Exhaustive Propagation] \label{exa:exhaustiveprop}
  Consider a $\BS(\LRA)$ clause set $N = \{x = 0 \parallel  \pnat(x),\; y = x + 1\parallel  \neg \pnat(x) \lor \pnat(y)\} \cup N'$
  where $N'$ is unsatisfiable and nothing can be propagated from $N'$.
  Let us further assume that $N'$ is satisfiable with respect to any instantiation
  of variables with natural numbers.
  If propagation is not restricted, then the first two clauses will consume
  all constants in $B$. For example, if $B = [a,b,c]$ then the trail
  $[\pnat(a), a=0, \pnat(b), b=a+1, \pnat(c), c=b+1]$ will be derived.
  Now all constants are fixed to natural numbers. So there cannot
  be a refutation of $N'$ anymore. An application of Grow will not
  solve the issue, because again the first two rules will fix all
  constants to natural numbers via exhaustive propagation.
\end{example}

\begin{definition}[Well-formed States]
  A state $(M;N;U;B;k;D)$ is \emph{well-formed}
  if the following conditions hold:
  \begin{enumerate}
  \item\label{wf-1}
    all constants appearing in $(M;N;U;B;k;D)$
    are from $B$ or occur in $N$.
  \item\label{wf-2}
    $M\land\adiff(B)$ is satisfiable
  \item\label{wf-3}
    $N\models_{\hspec} U$,
  \item Propagating clauses remain propagating and conflict clauses remain false: \label{wf-4}
    \begin{enumerate}
    \item \label{wf-4a}
      if $D = \Lambda\parallel C\cdot\sigma$ then $C\sigma$ is false in $\forgd(M)$
      and $\backgd(M)\land\adiff(B)\land\Lambda\sigma$ is satisfiable,
    \item \label{wf-4b}
      if $M = M_1,L\sigma^{(\Lambda\parallel C\lor L)\cdot\sigma},M_2$
      then $C\sigma$ is false in $\forgd(M_1)$,
      $L\sigma$ is undefined in $M_1$, and
      $\backgd(M_1)\land\adiff(B)\land\Lambda\sigma$ is satisfiable.
    \end{enumerate}
  \item All clauses in $N\cup U$ are pure. In particular, they don't contain any constants from $B$. \label{wf-5}
  \end{enumerate}
\end{definition}

\begin{lemma}[Rules preserve Well-Formed States] \label{lem:sclt-pres-wfs}
  The rules of $\SCLT$ preserve well-formed states.
\end{lemma}
\begin{proof}
  We prove each of the five properties by induction on the length
  of a derivation starting from the initial state
  $(\epsilon;N;\emptyset;B;k;\top)$. The induction step for
  the first two claims is
    \[(M;N;U;B;k;D)\Rightarrow_\SCLT(M';N';U';B';k';D').\]

  \medskip\noindent
  \ref{wf-1}.~In the initial state
  $(\epsilon;N;\emptyset;B;k;\top)$
  constants
  can only appear in $N$ and $B$,
  so it satisfies the claim.
  For the inductive step we do a case analysis on the rule application
  and prove $\con((M';N';U';B';k';D'))\subseteq\con(N')\cup B'$
  by case analysis on the rules of $\SCLT$.
  If we have applied Propagate or Decide then
  $N'=N$, $U'=U$, $B'=B$, $D'=D=\top$ and $M'=M,L\sigma$.
  So we only need to prove that $\con(M')\subseteq\con(N)\cup B$.
  Both rules require $|L\sigma| \in \atoms(\mGnd_B(N\cup U))$ satisfying the claim.

  In case of the rules Grow, Skip, or Backtrack, then
  $\con(N)\cup B\subseteq\con(N')\cup B'$ and
  $\con(M')\cup\con(U')\cup\con(D')\subseteq
  \con(M)\cup\con(U)\cup\con(D)$.

  If the rule Conflict was used then only the last component of
  the state changed $D'=\Lambda\parallel C\cdot\sigma$ with
  $\cdom(\sigma) \subseteq B$ and $\con(\Lambda\parallel C)\subseteq \con(N)\cup B$ by induction hypothesis.

  If one of the rules Resolve or Factorize were used
  then as for Conflict the only component of the state that
  changed was the conflict clause and the constants in $D'$
  are a subset of the constants in $M$ and $D$.

  \medskip\noindent
  \ref{wf-2}.~In the initial state $(\epsilon;N;\emptyset;B;k;\top)$
  the condition $\adiff(B)$ is satisfied as
  we assume constants in $B$ to be distinct.
  For the inductive step we do a case analysis on the rule application.
  If the rule used was one of Conflict, Backtrack, Skip, Resolve,
  or Factorize, then $M = M',M''$ with $M''$ possibly empty and
  $B=B'$, so that $M\land \adiff(B)$ satisfiable implies
  $M'\land \adiff(B')$ to be satisfiable.
  If the rule Grow was used then we have $M'=\epsilon$ and that all
  constants in $B'=B\oplus B''$ are distinct,
  so that satisfiability of $\adiff(B')$ is immediate.
  If the rule used was Propagate or Decide then we have
  $M'=M,L\sigma,\Lambda\sigma$ and $B' = B$, from the
  preconditions on the rules we also know that $L\sigma$ is undefined
  in $M$ and that $\backgd(M')\land\adiff(B')$ is satisfiable.

  \medskip\noindent
  \ref{wf-3}.~By induction on the number of learned clauses.
  We prove that for each application of Backtrack
  \[\begin{array}{l}
      (M,K^{i+1},M';N;U;B;k;D\lor L\cdot\sigma)\\
    \qquad \Rightarrow^{\text{Backtrack}}_\SCLT
      (M,L\sigma^{(\Lambda\parallel  D\lor L)\cdot\sigma},\Lambda'\sigma;N; U\cup\{\Lambda\parallel  D\lor L\};B;i;\top)
      \end{array}
  \]
  we have $N\cup U \models_\hspec D\lor L$. Following conflict
  resolution backward we can find a sequence of constrained closures
  $C_1\cdot\sigma_1,\dots,C_n\cdot\sigma_n$, where $C_n\cdot\sigma_n = D\lor L\cdot\sigma$ such that $C_1 \in (N\cup U)$ is the most recent conflict clause, and $C_{j+1}$ is either the result of a 
  factorization on $C_j$ or the result of a resolution inference between $C_j$ and a clause in $(N\cup U)$.
  By induction on the length of conflict resolution and soundness of resolution and factoring we get $N\cup U \models_\hspec D\lor L$.

  \medskip\noindent
  \ref{wf-4}.~For the initial state the properties \ref{wf-4a} and \ref{wf-4b}
  obviously hold. For the induction step and an application of the rules Decide, Skip, and Grow there
  is nothing to show.

  Consider a state $(M;N;U;B;k;\Delta\parallel D\cdot\delta)$
  obtained by an application of Conflict.
  By the side conditions of Conflict
  $\adiff(B)\land\backgd(M)\land\Delta\delta$ is satisfiable
  and $\forgd(M)\models\lnot(D\delta)$ is shown for \ref{wf-4a}.
  There is nothing to show for \ref{wf-4b}.
  
  Consider an application of rule Resolve
  \[
    \begin{array}{l}
    (M,L\rho^{\Lambda\parallel (C\lor L)\cdot\rho};N;U;B;k; \Lambda'\parallel (D\lor L')\cdot\sigma)\\
    \qquad \Rightarrow_\SCLT^{\text{Resolve}} (M,L\rho^{\Lambda\parallel (C\lor L)\cdot\rho};N;U;B;k;
      \Lambda\land\Lambda'\parallel (D\lor C)\eta\cdot\rho\sigma)\\
    \end{array}
  \]
  by induction hypothesis $\backgd(M)\land\adiff(B)\land\Lambda\eta\rho\sigma$ is satisfiable and $C\eta\rho\sigma$ is false in $\forgd(M)$,
  because $\Lambda\eta\rho\sigma = \Lambda\rho$ and $C\eta\rho\sigma = C\rho$ because $\eta$ is the mgu and we always assume clauses to be variable
  disjoint. Using the same argument $\backgd(M,L\sigma)\land\adiff(B)\land\Lambda'\eta\delta\sigma$ is satisfiable and $(D\lor L')\eta\delta\sigma)$ is false in $\forgd(M,L\sigma)$.
  Therefore $(D\lor C)\eta\cdot\rho\sigma$ is false in $\forgd(M,L\sigma)$ and $\backgd(M,L\sigma)\land\adiff(B)\land(\Lambda'\land \Lambda)\eta\delta\sigma$ is satisfiable, proving \ref{wf-4a}.
  There is nothing to show for \ref{wf-4b}.

  For an application of the rule Factorize there is nothing to show for \ref{wf-4b} and \ref{wf-4a} obviously holds because the set of
  different ground literals in $(D\lor L \lor L')\sigma$ and $(D\lor L)\sigma$ is identical.

  For an application of the rule Propagate there is nothing to show for \ref{wf-4a}. For \ref{wf-4b} consider the step
  \[
    \begin{array}{l}
      (M;N;U;B;k;\top)\\
      \qquad \Rightarrow_\SCLT^{\text{Propagate}} (M,L\sigma^{(\Lambda\parallel  C_0\lor L)\delta\cdot\sigma},\Lambda'\sigma;N;U;B;k;\top)
      \end{array}
    \]
  where the side conditions of the rule imply the claim modulo the removal of duplicate literals $L\sigma$.

      Finally, when applying Backtrack
      \[
    \begin{array}{l}
      (M,K^{i+1},M';N;U;B;k;(\Lambda\parallel  D\lor L)\cdot\sigma)\\
      \qquad \Rightarrow_\SCLT^{\text{Backtrack}} (M,L\sigma^{(\Lambda\parallel  D\lor L)\cdot\sigma},\Lambda'\sigma;N; U\cup\{\Lambda\parallel  D\lor L\};B;i;\top)
      \end{array}
      \]
      there is nothing to show for \ref{wf-4a}. For \ref{wf-4b} we know by induction hypothesis that $(D\lor L)\sigma$ is false in $\forgd(M,K^{i+1},M')$. The
      literal $L\sigma$ is of level $k$ and $D\sigma$ of level $i$, $k>i$, hence
      $D\sigma$ is false in $\forgd(M)$ and $L\sigma$ undefined in $\forgd(M)$. Furthermore, $\backgd(M,K^{i+1},M')\land\adiff(B)\land\Lambda\sigma$ is satisfiable
      by induction hypothesis, so $\backgd(M)\land\adiff(B)\land\Lambda\sigma$ is satisfiable as well.

      \medskip\noindent
      \ref{wf-4}.~For the initial state all clauses are pure by assumption. Conflict picks a clause from $N\cup U$ that is pure
      by induction hypothesis. Resolve and Factorize only apply unifiers between pure literals to the resulting clause, hence
      also only produce pure clauses from pure clauses. Finally, Backtrack adds the pure learned clause to $N\cup U$.
\end{proof}

\begin{definition}[Stuck State] \label{def:stuck-state}
  A state $(M;N;U;B;k;D)$ is called \emph{stuck} if $D\neq \Lambda\parallel \bot\cdot\sigma$ and none of the rules
  Propagate, Decide, Conflict, Resolve, Factorize,
  Skip, or Backtrack is applicable.
\end{definition}

\begin{proposition}[Form of Stuck States]
  If a run (without rule Grow) ends in a stuck state $(M;N;U;B;k;D)$ where Conflict was applied eagerly, then $D=\top$ and all ground foreground
  literals that can be build from the foreground literals in $N$ by instantiation
  with constants from $B$ are defined in $M$.
\end{proposition}
\begin{proof}
  First we prove that stuck states never appear during conflict resolution.
  Consider a well-formed state $(M;N;U;B;k;\Delta\parallel D\cdot\delta)$,
  we prove by case analysis that either Skip, Resolve, Factorize or Backtrack can be applied.
  If $M = M', L\sigma$ and $L\sigma$ is either a background literal or a foreground literal such that
  $\comp(L\sigma)$ is not contained in $D\delta$ then Skip can be applied.
  If $M = M', L\sigma^{\Lambda\parallel C\cdot\sigma}$ with $D\delta = D'\lor \comp(L\sigma)$ then
  Resolve can be applied.
  If $M = M',L\sigma^k, M''$ and $D'$ contains multiple occurrences of $\comp(L\sigma)$ then Factorize can be applied.
  In summary, we can reach a state with a unique literal $L\delta$ of level $k$ in $D\delta$.
  Then Backtrack is applicable.
  Finally, if in some state $(M;N;U;B;k;\top)$ where Conflict is not applicable, some atom $|L|\in\atoms(\mGnd_B(N))$ is undefined, we can always apply Decide.
\end{proof}

\begin{lemma}[Stuck States Produce Ground Models] \label{lem:stuck-models}
  If a state $(M;N;U;B;k;\top)$ is stuck then $M\land\adiff(B)\models \mGnd_B(N\cup U)$.
\end{lemma}
\begin{proof}
  By contradiction. Note that $M\land\adiff(B)$ is satisfiable, Lemma~\ref{lem:sclt-pres-wfs}.\ref{wf-2}.
  Consider any clause $(\Lambda\parallel C)\sigma\in \mGnd_B(N\cup U)$.
  It can only be not true in $M\land\adiff(B)$ if $\forgd(M) \models \neg (C\sigma)$
  and $\backgd(M)\land\adiff(B)\land\Lambda\sigma$ is satisfiable. But then Conflict would be applicable,
  a contradiction.
\end{proof}

\begin{example}[$\SCLT$ Model Extraction] \label{exa:modelextract}
  In some cases it is possible to extract an overall model from the ground trail of a stuck state of an $\SCLT$ derivation.
  Consider $B=[a,b,c]$ and a satisfiable $\BS(\LRA)$ constrained clause set
  $N = \{x\ge 1\parallel P(x),\allowbreak x<0\parallel P(x),\allowbreak 0\le x\land x<1\parallel \lnot P(x),\allowbreak 2x\ge 1\parallel P(x)\lor Q(x)\}$.
  Starting from state $(\epsilon;N;\emptyset;B;0;\top)$ and applying Propagate fairly a regular run can derive the following trail\newline
  \renewcommand{\arraystretch}{1.2}
  $\begin{array}{rl}
      M = & P(a)^{x\ge 1\parallel P(x)\cdot\{x\mapsto a\}}, a\ge 1, P(b)^{x<0\parallel P(x)\cdot\{x\mapsto b\}}, b<0,\\
         & \lnot P(c)^{0\le x\land x<1\parallel \lnot P(x)\cdot\{x\mapsto c\}},0\le c,c<1,
       Q(c)^{2x\ge1\parallel P\lor Q(x)\cdot\{x\mapsto c\}}, 2c\ge 1 \\
    \end{array}
    $

  \noindent  
  The state $(M;N;\emptyset;B;0;\top)$ is stuck and $M\models_\hspec\mGnd_B(N)$.
  Moreover from $M$ we can generate an interpretation $\sigval^{\BS(\LRA)}$ of $N$
  by generalizing the foreground constants used for instantiation and interpreting
  the predicates $P$ and $Q$ as formulas over $\sig^\bspec$, $P^{\sigval} = \{q\in\Rat \mid q < 0 \lor q \ge 1\}$
  and $Q^{\sigval} = \{q\in\Rat \mid 2q \ge 1 \land q < 1\}$.
\end{example}

\begin{lemma}[Soundness] \label{lem:sclsound}
  If a derivation reaches the state $(M;N;U;B;k;\Lambda\parallel \bot\cdot\sigma)$,
  then $N$ is unsatisfiable.
\end{lemma}
\begin{proof}
  All learned clauses are consequences of $N\cup U$, Lemma~\ref{lem:sclt-pres-wfs}.\ref{wf-3}.
  Furthermore $\backgd(M)\land\adiff(B)\land\Lambda\sigma$ is satisfiable, Lemma~\ref{lem:sclt-pres-wfs}.\ref{wf-4a}.
\end{proof}

% \begin{definition}[Fair Run] \label{def:fair-run}
%   \Todo{Remove this old definition.}
%   A sequence of $\SCLT$ rule applications is called a \emph{fair run}
%   if from the initial state or after the application of rule Decide the rule Propagate
%   is: (i)~applied once to all applicable clauses if the size of $B$
%   suffices and no conflict is generated
%   (ii)~applied a second time to the same clause only if this generates
%   a conflict.
% \end{definition}

\begin{definition}[Reasonable Run] \label{def:reasonable-run}
  A sequence of $\SCLT$ rule applications is called a \emph{reasonable run}
  if the rule Decide is only applied if there exists no application of the rule Propagate that would generate a conflict.
\end{definition}

%  \Todo{Note to myself (MB): not needed for termination or stuck states (there are other solutions), but to guarantee non-redundancy!}

% \begin{definition}[Progressing Run] \label{def:progressing-run}
%   A sequence of $\SCLT$ rule applications is called a \emph{progressing run}
%   if the rule Decide is only applied 
%   if the rule Propagate has been applied at least once to all applicable clauses 
%   since the last application of Decide or the the initial state.
  
%   \Todo{Alternative wording:
%   A sequence of $\SCLT$ rule applications is called a \emph{progressing run}
%   if the rule Decide is only applied 
%   when there is no clause to which the rule Propagate can be applied and 
%   that has not been propagated since the initial state or the last application of Decide.}
% \end{definition}

% \begin{definition}[Fair Run] \label{def:fair-run}
%   A sequence of $\SCLT$ rule applications is called a \emph{fair run}
%   if the rule Propagate is applied at most once to all applicable clauses 
%   before the first application of the rule Decide and in-between two subsequent applications of the rule Decide.   
%   The only exception is if the second application to the same clause generates a conflict.
% \end{definition}

\begin{definition}[Regular Run] \label{def:regular-run}
  A sequence of $\SCLT$ rule applications is called a \emph{regular run} if
  it is a reasonable run the rule Conflict has precedence over all
  other rules, and Resolve resolves away at least the rightmost
  foreground literal from the trail.
\end{definition}

\begin{example}[$\SCLT$  Refutation]
  Given a set of foreground constants $B=[a,b,c]$ and a $\BS(\LRA)$ constrained clause set  $N=\{
  C_1\colon x = 0 \parallel  P(x), \allowbreak
  C_2\colon  y= x + 1 \parallel \lnot P(x)\lor P(y),
  C_3\colon  z = 2 \parallel  \lnot P(z)\}$
  the following is a regular derivation
  \renewcommand{\arraystretch}{1.2}
  \[
    \begin{array}{ll}
      &(\epsilon;N;\emptyset;B;0;\top)\\
      \Rightarrow^{\text{Propagate}}_\SCLT&(P(a)^{C_1\cdot\{x\mapsto a\}},a=0;N;\emptyset;B;0;\top)\\
      \Rightarrow^{\text{Propagate}}_\SCLT&(\ldots,P(b)^{C_2\cdot\{x\mapsto a, y\mapsto b\}},b= a + 1;N;\emptyset;B;0;\top)\\
      \Rightarrow^{\text{Propagate}}_\SCLT&(\ldots,P(c)^{C_2\cdot\{x\mapsto b, y\mapsto c\}},c= b + 1;N;\emptyset;B;0;\top)\\
      \Rightarrow^{\text{Conflict}}_\SCLT&(\ldots,  P(c)^{C_2\cdot\{x\mapsto b, y\mapsto c\}},c= b + 1;N;\emptyset;B;0; z = 2 \parallel  \lnot P(z) \cdot\{z\mapsto c\})\\
      \Rightarrow^{\text{Resolve}}_\SCLT&(\ldots,P(c)^{C_2\cdot\{x\mapsto b, y\mapsto c\}},c= b + 1;N;\emptyset;B;0;\\
      &z= x + 1 \land z = 2\parallel \lnot P(x)\cdot\{z\mapsto c,x\mapsto b\})\\
      \Rightarrow^{\text{Skip}}_\SCLT&(\ldots,P(b)^{C_2\cdot\{x\mapsto a, y\mapsto b\}},b= a + 1;N;\emptyset;B;0;\\
      &z= x + 1 \land z = 2\parallel \lnot P(x)\cdot\{z\mapsto c,x\mapsto b\})\\
      \Rightarrow^{\text{Resolve}}_\SCLT&(\ldots,P(b)^{C_2\cdot\{x\mapsto a, y\mapsto b\}},b= a + 1;N;\emptyset;B;0;\\
      &z= x + 1 \land z = 2 \land x= x_1 + 1\parallel \lnot P(x_1)\cdot\{z\mapsto c,x\mapsto b, x_1\mapsto a\})\\
      \Rightarrow^{\text{Skip}}_\SCLT&(P(a)^{C_1\cdot\{x\mapsto a\}},a= 0;N;\emptyset;B;0;\\
      &z= x + 1 \land z = 2 \land x= x_1 + 1\parallel \lnot P(x_1)\cdot\{z\mapsto c,x\mapsto b, x_1\mapsto a\})\\
      \Rightarrow^{\text{Resolve}}_\SCLT&(P(a)^{C_1\cdot\{x\mapsto a\}},a= 0;N;\emptyset;B;0;\\
      &z= x + 1 \land z = 2 \land x= x_1 + 1\land x_1 = 0 \parallel \bot\cdot\{z\mapsto c,x\mapsto b, x_1\mapsto a\})\\
    \end{array}
  \]
  $N$ is proven unsatisfiable as we reach a state in the form $(M;N;U;B;k;\Lambda\parallel \bot\cdot\sigma)$.
\end{example}

\begin{example}[$\SCLT$ Clause learning]
  Given an initial constant set $B = [a]$ of fresh foreground constants
  and a $\BS(\LRA)$  constrained clause set
  $ N =\{
  C_1\colon            x\ge y\parallel  \lnot P(x,y) \lor Q(z),\allowbreak
  C_2\colon          z = u+v \parallel \lnot P(u,v) \lor\lnot Q(z),\allowbreak
  \}$
  the following is an example of a regular run\newline
  \renewcommand{\arraystretch}{1.2}
  \[
    \begin{array}{ll}
      &(\epsilon;N;\emptyset;B;0;\top)\\
      \Rightarrow_\SCLT^{Decide}&(P(a,b)^1;N;\emptyset;B;1;\top)\\
      \Rightarrow_\SCLT^{Propagate}&(P(a,a)^1,Q(a)^{C_1\cdot\{x\mapsto a,y\mapsto a,z\mapsto a\}},a\ge a;N;\emptyset;B;1;\top)\\
      \Rightarrow_\SCLT^{Conflict}&(P(a,a)^1,Q(a)^{C_1\cdot\{u\mapsto a,v\mapsto a,z\mapsto a\}},a\ge a;N;\emptyset;B;1;\\
      &C_2\cdot\{x\mapsto a,y\mapsto a,z\mapsto a\})\\
      \Rightarrow_\SCLT^{Resolve}&(P(a,a)^1,Q(a)^{C_1\cdot\{x\mapsto a,y\mapsto a,z\mapsto a\}},a\ge a;N;\emptyset;B;1;x\ge y \land z= u+v\parallel \\
      &\lnot P(x,y)\lor\lnot P(u,v)\cdot\{x\mapsto a,y\mapsto a,z\mapsto a,u\mapsto a,v\mapsto a\})\\
      \Rightarrow_\SCLT^{Skip*}&(P(a,a)^1;N;\emptyset;B;1;x\ge y\land z= u+v\parallel \\
      &\lnot P(x,y)\lor\lnot P(u,v)\cdot\{x\mapsto a,y\mapsto a,z\mapsto a,u\mapsto a,v\mapsto a\})\\
      \Rightarrow_\SCLT^{Factorize}&(P(a,a)^1;N;\emptyset;B;1;x\ge y\land z= x+y\parallel  \lnot P(x,y)\cdot\{x\mapsto a,y\mapsto a,z\mapsto a\})\\
      \Rightarrow_\SCLT^{Backtrack}&(\lnot P(a,a)^{(x\ge y\land z= x+y\parallel  \lnot P(x,y))\cdot\{x\mapsto a,y\mapsto a\}}, a\ge a, a= a+a;N;\\
      &\{x\ge y\land z= x+y\parallel  \lnot P(x,y)\};B;1;\top)\\
    \end{array}
  \]
  \noindent
  In this example the learned clauses $x\ge y\land z= x+y\parallel  \lnot P(x,y)$; note how there are two distinct variables in the learned clause even if we had to use a single constant for instantiations in conflict search.
\end{example}

%%%%%%%%%%%%%%%%%%% -%%%%%%%%%%%%%%%%%%%-%%%%%%%%%%%%%%%%%%%-%%%%%%%%%%%%%%%%%%%

\begin{proposition}\label{prop:no-decide-conflict}
  Let $N$ be a set of constrained clauses.
  Then any application of Decide in an $\SCLT$ regular run from starting state $(\epsilon;N;\emptyset;B;0;\top)$ does not create a conflict.
\end{proposition}
\begin{proof}
  Assume the contrary: then Propagate would have been applicable before Decide, contradicting with
  the definition of a regular and hence reasonable run.
\end{proof}

\begin{corollary}\label{corol:resolve-after-conflict}
Let $N$ be a set of constrained clauses.
Then any conflict in an $\SCLT$ regular run from starting state $(\epsilon;N;\emptyset;B;0;\top)$ admits a regular conflict resolution.
\end{corollary}
\begin{proof}We need to prove that it is possible to apply Resolve during conflict resolution. By Proposition~\ref{prop:no-decide-conflict} the rightmost foreground literal
  on the trail is a propagation literal and by regularity we know that this literal appears in the conflict clause. So a conflict resolution can start by skipping over the
  background literals and then resolving once with the rightmost foreground literal.
\end{proof}

\begin{lemma}[Non-Redundant Clause Learning]\label{lemm:non-red}
  Let $N$ be a set of constrained clauses.
  Then clauses learned in an $\SCLT$ regular run from starting state $(\epsilon;N;\emptyset;B;0;\top)$ are not
  redundant.
\end{lemma}
\begin{proof}
  Consider the following fragment of a derivation learning a clause:\newline
  \renewcommand{\arraystretch}{1.2}
  \centerline{$\begin{array}{ll}
    \Rightarrow^{\text{Conflict}}_{\SCLT} & (M'';N;U;B;k;\Lambda_0\parallel C_0\cdot\sigma_0)\\
    \Rightarrow^{\{\text{Skip, Factorize, Resolve}\}^*}_{\SCLT} & (M,K^{i+1},M';N;U;B;k;\Lambda_n\parallel C_n\cdot\sigma_n)\\
    \Rightarrow^{\text{Backtrack}}_{{\SCLT}} & (M,L\sigma^{(\Lambda_n\parallel  D\lor L)\cdot\sigma},\Lambda_n'\sigma;N;U\cup\{\Lambda_n\parallel D\lor L\};B;i;\top).\\
   \end{array}$}

 \noindent
 where $C_n = D\lor L$ and $\sigma = \sigma_n$.
  Let $\prec$ be any $\hspec$ order induced by $M$. We prove that $\Lambda_n\sigma\parallel C_n\sigma$
  is not redundant with respect to $\prec$, $B$, and $(N\cup U)$.
  By soundness of hierarchic resolution $(N\cup U) \models \Lambda_n\parallel C_n$
  and $\Lambda_n\sigma$ is satisfiable with $M\land\adiff(B)$, and $C_n\sigma$ is false under both $M$ and $M,K^{i+1},M'$, Lemma~\ref{lem:sclt-pres-wfs}.
  For a proof by contradiction, assume there is a
  $N'\subseteq\mGnd_B(N\cup U)^{\preceq \Lambda_n\sigma\parallel C_n\sigma}$
  such that $N'\models_\hspec\Lambda_n\sigma\parallel C_n\sigma$.
   As $\Lambda_n\sigma\parallel C_n\sigma$ is false under $M$,
  there is a ground constrained clause $\Lambda'\parallel C'\in N'$ with $\Lambda'\parallel C' \preceq \Lambda_n\sigma\parallel C_n\sigma$,
  and all literals from $C'$ are defined in $M$ and false by the definition of $\prec$.
  Furthermore, we can assume that $\adiff(B)\land\backgd(M)\land\Lambda'$ is satisfiable or 
  $C_n\sigma$ would be a tautology, because $\adiff(B)\land\backgd(M)\land\Lambda_n\sigma$ is satisfiable. 
  
  The clause $\Lambda_0\sigma_0\parallel C_0\sigma_0$ has at least one literal of level $k$ and due to a regular run, Definition~\ref{def:regular-run},
  the rightmost trail literal is resolved away in $\Lambda_n\sigma\parallel C_n\sigma$, Corollary~\ref{corol:resolve-after-conflict}. Therefore, the rightmost foreground literal does not appear in
  $\Lambda'\parallel C'$, so by regularity $\Lambda'\parallel C'$ would have created a conflict at a previous state.
\end{proof}

Of course, in a regular run the ordering of foreground literals on the trail will change, i.e., the ordering
underlying Lemma~\ref{lemm:non-red} will change as well. Thus the non-redundancy property of Lemma~\ref{lemm:non-red}
reflects the situation at the time of creation of the learned clause. A non-redundancy property holding for
an overall run must be invariant against changes on the ordering. However, the ordering underlying Lemma~\ref{lemm:non-red} also entails 
a fixed subset ordering that is invariant against changes on the overall ordering. This means that our dynamic ordering entails non-redundancy criteria based on subset relations including forward redundancy.
From an implementation perspective, 
this means that learned clauses need
not to be tested for forward redundancy. Current resolution, or superposition based provers spent a reasonable portion
of their time in testing forward redundancy of newly generated clauses. In addition, also tests for backward reduction can be
restricted knowing that learned clauses are not redundant.

\begin{lemma}[Termination of $\SCLT$]\label{thm:termination}
  Let $N$ be a set of constrained clauses and $B$
  be a finite set of background constants.
  Then any regular run with start state
  $(\epsilon;N;\emptyset;B;0;\top)$
  that uses Grow only finitely often terminates.
\end{lemma}
\begin{proof}
  Since Grow can only be used a finite number of times
  we consider as a start state the state after the final application
  of Grow and prove termination of runs that never use Grow.
  We do so by giving an explicit termination measure on the $\SCLT$ states.
  Given a state $(M;N;U;B;k;D)$ we define a termination measure $\mu$ as
  $\mu(M;N;U;B;k;D) = (u,s,m,r,d) \in \mathbb N^5$ with a lexicographical combination of $>$
  where
  \begin{itemize}
  \item $l=|\atoms(\mGnd_B(N\cup U))|$, $u= 3^l - |\mGnd_B(U)|$, and $m= |M|$,
  \item in the case $D=\top$:
    \begin{itemize}
    \item[\textasteriskcentered] $s = 1+l-m$, $d=0$, and $r=0$,
    \end{itemize}
  \item otherwise if $D=\Delta\parallel D'\cdot\delta$:
    \begin{itemize}
    \item[\textasteriskcentered] $s=0$,
      % \begin{itemize}
    \item[\textasteriskcentered] if $M=M',L$ with $L$ foreground literal then $r$ is the number of copies of $L$ in $D'\delta$
    \item[\textasteriskcentered] if the rightmost literal of $M$ is a background literal or if $M$ is empty then $r=0$
    \item[\textasteriskcentered] $d$ is the number of literals in $D'$
      % \end{itemize}
    \end{itemize}
  \end{itemize}
  %%%%% 
  % where $l={\atoms(\mGnd_B(N\cup U))}$; $m= |M|$; $u= 3^l - |\mGnd_B(U)|$; and the
  % three components
  % $s$, $d$, and $r$ are defined conditionally on $D$ as follow:
  % if $D=\top$ then $s = 1+l-m$, $d=0$, and $r=0$;
  % if $D=\Delta\parallel D'\cdot\delta$ then $s=0$, $d=|D'|$,
  % and $r$ is the number of copies in $D'\delta$ of the rightmost
  % foreground literal on the trail $M$; if the trail does not contain any foreground literals then $r=0$.
  % $\mu(M;N;U;B;k;D) = (u,s,m,r,d) \in \mathbb N^5$ with a lexicographical combination of $>$
  %%%%%
  % where $l={\atoms(\mGnd_B(N\cup U))}$;
  % $u= 3^l - |\mGnd_B(U)| \in \mathbb N$;
  % $m= |M|$;
  % $s = 1+l-m$ if $D=\top$ otherwise $s=0$;
  % if $M=M',L$ and $D=\Delta\parallel D'\cdot\delta$
  % then $d = |D'|$ and is $r$
  % the number of copies of $L$ in $D'\delta$ and if $D=\top$
  % then $s = 0 =d$.
  The number of ground atoms $l=|\atoms(\mGnd_B(N\cup U))|$ is an upper bound to the length of the
  trail because the trail is consistent and no literal can appear
  more than once on the trail. 
  Similarly, every learned clause has at least one non-redundant
  ground instance so $|\mGnd_B(U)|$ increases whenever SCL(T) learns a new clause and 
  $3^l$ is an upper bound to the ground instances of all learned clauses in a regular run. 
  This means that Backtrack strictly decreases $u$, Decide, Propagate,
  and Conflict strictly decrease $s$ without modifying $u$,
  Skip strictly decreases $m$ without modifying $u$ or $s$,
  Resolve strictly decreases $r$ without modifying $u$, $s$, or $m$,
  and finally Factorize strictly decreases $d$ possibly decreases $r$ and
  does not modify $u$, $s$, or $m$.
\end{proof}

\begin{theorem}[Hierarchic Herbrand Theorem] \label{theo:hierherbrand}
  Let $N$ be a finite set of clauses.
  $N$ is unsatisfiable iff there exists a finite set $N' = \{ \Lambda_1 \parallel C_1, \ldots, \Lambda_n \parallel C_n\}$ of variable renamed
  copies of clauses from $N$ and a finite set $B$ of fresh constants and a substitution $\sigma$, grounding for $N'$ where
  $\cdom(\sigma) = B$ such that $\bigwedge_i \Lambda_i\sigma$ is $\theory^{\bspec}$ satisfiable and  $\bigwedge_i C_i\sigma$ is first-order unsatisfiable
  over $\sig^{\fspec}$.
\end{theorem}
\begin{proof}
  Recall that $N$ is a pure, abstracted clause set and that $\theory^{\bspec}$ is term-generated, compact background theory that contains
  an equality $=$, and that all constants of the background signature are domain constants.
  Then by completeness of hierarchic superposition~\cite{BachmairGanzingerEtAl94}, $N$ is unsatisfiable iff
  there exists a refutation by hierarchic superposition. Let $N' = \{ \Lambda_1 \parallel C_1, \ldots, \Lambda_n \parallel C_n\}$ be a finite set renamed copies of clauses from $N$ such
  that there is a refutation by hierarchic superposition such that each clause in $N'$ and each derived clause is used exactly once.
  This set exists because the refutation is finite and any hierarchic superposition refutation can be transformed
  into a refutation where every clause is used exactly once. Now let $\delta$ be the overall unifier of this refutation.
  This unifier exists, because all clauses in $N'$ have disjoint variables and all clauses in the refutation are used exactly once.
  Now we consider a finite set of constants $B$ and a substitution $\sigma$,  $\cdom(\sigma) = B$, $\sigma$ grounding for $N'$, and
  for all $x,y\in\dom(\delta)$ we have $x\sigma = y\sigma$ iff $x\delta = y\delta$ . Now there is also a refutation for $N'\sigma$
  by hierarchic superposition where the clauses are inferred exactly in the way they were inferred for $N'$. It remains to be shown
  that  $\bigwedge_i \Lambda_i\sigma$ is $\theory^{\bspec}$ satisfiable and $\bigwedge_i C_i\sigma$ is $\sigval^{\hspec}$ unsatisfiable.
  The hierarchic superposition refutation terminates with the clause $\bigwedge_i \Lambda_i\sigma \parallel \bot$ where $\bigwedge_i \Lambda_i\sigma$
  is satisfiable. Furthermore, the refutation derives $\bot$ from $\{ C_1\sigma, \ldots, C_n\sigma\}$ via superposition, proving the theorem. 
\end{proof}

% \begin{theorem}[Ground Refutational completeness] \label{theo:ground-model-comple}
%   For some finite fixed set $B$ where $\mGnd_B(N)$ is unsatisfiable
%   any regular $\SCLT$ run without rule Grow will
%   derive the empty clause:\newline
%   \centerline{$(\epsilon;N;\emptyset;B;0;\top)
%     \Rightarrow_\SCLT^*
%     (M;N;U;B;k;\Lambda\parallel \bot\cdot\sigma)$}
% \end{theorem}
% \begin{proof}
%   Consider a regular run starting from $(\epsilon;N;\emptyset;\emptyset;B;0;\top)$, by Lemma~\ref{thm:termination} all regular runs terminate so the final state needs to
%   be either a stuck state $(M;N;B;k;\top)$ or in the form $(M;N;U;B;k;\Lambda\parallel \bot\cdot\sigma)$.
  
%   Assume by contradiction that the final state is a stuck state $(M;N;U;B;k;\top)$, then we have that $M$ is satisfiable and by Lemma~\ref{lem:stuck-models}
%   we also have $M\models\mGnd_B(N)$, in contradiction with the unsatisfiability of $\mGnd_B(N)$.
% \end{proof}

Finally, we show that an unsatisfiable clause set can be refuted by $\SCLT$ with any regular run if we start
with a sufficiently large sequence of constants $B$ and apply Decide in a fair way. In addition, we need a Restart
rule to recover from a stuck state.

\bigskip
\shortrules{Restart}
{$(M;N;U;B;k;\top)$}
{$(\epsilon;N; U;B;0;\top)$}
%{$(M,\Lambda\sigma,L\sigma^{\Lambda\parallel D\lor L\cdot\sigma};N;
%  U\cup\{\Lambda\parallel  D'\lor L\};B;i;\top)$}
{}{\SCLT}{11}

Of course, an unrestricted use of rule Restart immediately leads to non-termination.

\begin{theorem}[Refutational Completeness of $\SCLT$] \label{theo:groununsatcomple}
  Let $N$ be an unsatisfiable clause set.
  Then any regular $\SCLT$ run will
  derive the empty clause provided
  (i)~Rule Grow and Decide are operated in a fair way, such that all possible trail prefixes of all considered sets $B$ during the run are eventually explored, and
  (ii)~Restart is only applied to stuck states.
\end{theorem}
\begin{proof}
  If $N$ is unsatisfiable then by Theorem~\ref{theo:hierherbrand} there exists a
  a finite set $N' = \{ \Lambda_1 \parallel C_1, \ldots, \Lambda_n \parallel C_n\}$ of variable renamed
  copies of clauses from $N$ and a finite set $B$ of fresh constants and a substitution $\sigma$, grounding for $N'$ where
  $\cdom(\sigma) = B$ such that $\bigwedge_i \Lambda_i\sigma$ is $\theory^{\bspec}$ satisfiable and  $\bigwedge_i C_i\sigma$ is first-order unsatisfiable
  over $\sig^{\fspec}$. If the $\SCLT$ rules are applied in a fair way, then they will in particular produce trails solely consisting of literals from $N'\sigma$.
  For these trails all theory literals are satisfiable, because $\bigwedge_i \Lambda_i\sigma$ is $\theory^{\bspec}$ satisfiable.
  Furthermore, the states corresponding to these trails cannot end in a stuck state, because this contradicts the unsatisfiability of $\bigwedge_i C_i\sigma$.
  Instead, they all end in a conflict with some clause in $N'\sigma$.
  In addition, there are only finitely many such trails, because the number of literals in $N'\sigma$ is finite.
  Now let $\mu((M;N;U;B;k;\top))$ be the multiset of the levels of all states with trails from $N'\sigma$ until a conflict occurs.
  Each time a state with a trail from $N'\sigma$ results in a conflict, $\SCLT$ learns a non-redundant clause that propagates at a strictly smaller level, Lemma~\ref{lemm:non-red}.
  Thus $\mu((M;N;U;B;k;\top))$ strictly decreases after each Backtrack step after a conflict on a trail with atoms from $N'\sigma$.
  The clause learnt at level zero is the empty clause.
\end{proof}

Condition~(i) of the above theorem is quite abstract. It can, e.g., be made effective by applying rule Grow only after all possible trail prefixes with respect to the current set $B$ have been explored and to make sure that Decide does not produce the same stuck state twice.

%%% Local Variables:
%%% mode: latex
%%% TeX-master: "paper"
%%% End:

% LocalWords:  cardinality postconditions

\section{$\SCLT$ Extensions and Discussion} \label{sec:extdisc}

Example~\ref{exa:modelextract} demonstrates that stuck $\SCLT$ states can sometimes
be turned into models of the considered clause set. There exist decidable fragments where this
can be done systematically. Consider a pure $\BS(\LRA)$ clause set
$N$ where arithmetic atoms have the form $x \# q$, $x \# y$ or $x - y \# q$
where $q$ is a domain constant and $\#\in \{\leq, <, \not =, =, >, \geq\}$. Furthermore, if
a clause contains a difference constraint $x - y \# q$, then $x$ and $y$ are
in addition bounded from below and above in the respective clause constraint. The resulting fragment is called BS(BD)~\cite{Voigt17Frocos} and satisfiability
of clause sets in this fragment is decidable. This can be effectively done by searching for a refutation with respect to an \textit{a priori}
fixed set of constants $B$. The constants in $B$ are all different, ordered, and enjoy additional bounds with respect to the domain
constants occurring in $N$~\cite{Voigt17Frocos}.
The clause set is then unsatisfiable iff it is unsatisfiable after instantiation with $B$. Otherwise
a model can be generated, i.e., BS(BD) enjoys an abstract finite model property.
The existence of a refutation can be checked by $\SCLT$ where the
additional bounds on the constants in $B$ can be introduced with the first decision.
Actually, \cite{Voigt17Frocos} was
part of our motivation for the model building in $\SCLT$.

There are further extensions to pure clause sets that still enable
a refutationally complete calculus. In particular, first-order function symbols that do not range
into a background theory sort and equality. The properties of the $\SCLT$
calculus rely on finite trails with respect to a fixed, finite set $B$ of constants.
By adding non-constant first-order function symbols trails will typically
be infinite without further restrictions. Finite trails can, e.g., still be obtained
by limiting nestings of function symbols in terms. Thus it seems to us that an
extension to first-order function symbols that do not range
into a background theory sort should be possible while keeping the properties of $\SCLT$.
From an abstract point of view, also the addition of equality on the first-order
side should be possible, because there exist complete procedures such as hierarchic superposition~\cite{BachmairGanzingerEtAl94,BaumgartnerWaldmann19}.
Then also foreground function symbols may range into a background theory sort, but the respective terms have to satisfy further conditions in order
to preserve completeness.
However, even in the pure first-order case there has not been a convincing solution so
far of how to combine equational reasoning with explicit model building. One challenge
is how to learn a clause from a conflict out of a partial model assumption that enjoys superposition style ordering restrictions
on terms occurring in equations. If this can be sufficiently solved, the respective calculus should also be extendable to a hierarchic set up.

An efficient implementation of $\SCLT$ requires efficient algorithmic solutions
to a number of concepts out of the theory. For fast model building an efficient
implementation of Propagate is needed. This was our motivation for adding the all-different
constraints on the constants, because they enable syntactic testing for complementary or defined
literals. In addition, satisfiability of constraints needs to be tested. The trail behaves like a stack
and it is ground. This
fits perfectly the strengths of SMT-style satisfiability testing.
Dealing with the non-domain constants out of the set $B$ needs some care. They behave completely symmetric with
respect to the instantiation of clauses in $(N\cup U)$. An easy way to break symmetry here is the addition of
linear ordering constraints on these constants. If more is known about the specific fragment some clause set $N$
belongs to, additional constraints with respect to the constraints or domain constants out of $(N\cup U)$ may
be added as well. This is for example the case for the BS(BD) fragment, see above.
Exploring all trail prefixes, as required by Theorem~\ref{theo:groununsatcomple}, requires
book-keeping on visited stuck states and an efficient implementation of the rule Restart. The former can be done
by actually learning new clauses that represent stuck states. Such clauses are not logical consequences out
of $N$, so they have to be treated special. 
In case of an application of Grow all these clauses and all the consequences
thereof have to be updated. An easy solution would be to forget the clauses generated by stuck states. This can be efficiently implemented.
Concerning the rule Restart, from the SAT world it is known that restarts do not have to be total~\cite{RamosEtAl11}, i.e., if a certain
prefix of a trail will be reproduced after a restart, it can be left on the trail. It seems possible to extend this concept towards $\SCLT$.

There seems to be the possibility to learn more from a stuck state than just a clause preventing it for the current run until
the next application of Grow. For example, if it turns out that the finite model out of a stuck state cannot be extended to an overall
model, it might trigger some non-redundant inference, similar to the InstGen calculus~\cite{GanzingerKorovinEtAl03}, or result in further constraints for
the constants out of $B$, or actually show that Grow needs to be applied in order to find a model. Finding answers to these
problems is subject to future research.

It was enough to restrict $\SCLT$ to regular runs, Definition~\ref{def:regular-run}, to prove the termination of our calculus. 
An actual implementation might benefit from further restrictions to the explored sequences of rule applications. 
For SAT solving, one of those restrictions is a focus on greedy propagation. 
The same might not be true for $\SCLT$. 
Still, we assume that a restriction that enforces some progress through propagation might be necessary, e.g., the rule Decide is only applied 
if the rule Propagate has been applied at least once to all applicable clauses since the last application of Decide or the the initial state. 
We also assume that some kind of fairness guarantee must be kept to prevent a small subset of clauses from causing the majority of propagated literals.
But these are for sure not the only options.

In summary, we have presented the new calculus $\SCLT$ for pure clause sets of first-order logic modulo a background theory.
The calculus is sound and complete. It does not generate redundant clauses. There is a fair strategy that turns it into
a semi-decision procedure. It constitutes a decision procedure for certain decidable fragments of pure clause sets, such as BS(BD).

%%% Local Variables:
%%% mode: latex
%%% TeX-master: "paper"
%%% End:

%\input{paper-dpll}

%\input{paper-examples}

\smallskip\noindent
{\bf Acknowledgments:} This work was funded by DFG grant 389792660 as part of
\href{http://perspicuous-computing.science}{TRR~248}.

\bibliographystyle{alpha}
\bibliography{paper}

\end{document}